\documentclass[10pt]{article}
\usepackage{amsmath,amsfonts}
\usepackage{amssymb,amscd,amsthm}
\usepackage[all]{xy}
\usepackage[dvips]{graphicx}
\usepackage{verbatim}
\usepackage[perpage,symbol*]{footmisc}
\usepackage[justification=centering]{caption}
\usepackage{longtable}
\usepackage{multirow}
\usepackage{slashbox}
\usepackage[figuresright]{rotating}
\usepackage{makecell}
\usepackage{rotating}
\newtheorem{lemma}{\qquad\textbf{Lemma}}
\newtheorem{theorem}{\qquad\textbf{Theorem}}

\newtheorem{corollary}{\qquad\textbf{Corollary}}

\newtheorem{example}{\qquad\textbf{Example}}

\newcommand{\F}{\mathbb{F}}

\newcommand{\oa}{\overrightarrow{\alpha}}
\newcommand{\ov}{\overrightarrow{v}}
\begin{document}

\baselineskip 17pt
\title{\Large\bf Two Classes of Quantum MDS Codes with Large Minimal Distance}
\author{\large  Puyin Wang \quad\quad Jinquan Luo*}\footnotetext{The authors are with School of Mathematics
and Statistics \& Hubei Key Laboratory of Mathematical Sciences, Central China Normal University, Wuhan China 430079.\\
 E-mail: py.wang1998@foxmail.com(P.Wang), luojinquan@ccnu.edu.cn(J.Luo)}
\date{}
\maketitle

{\bf Abstract}:
   In this paper, two classes of quantum MDS codes are constructed. The main tools are multiplicative structures on finite fields. Carefully choosing different cosets can make the corresponding  generalized Reed-Solomon codes Hermitian self-orthogonal, which yields quantum MDS codes by utilizing Hermitian type CSS construction.  In some cases, these codes have minimal distances larger than $q/2$. Also, the parameters of these codes have never been reported before.


\section{Introduction}
\qquad Since the mid-20th century, quantum computing has gradually emerged as a pivotal field in communication, attracting widespread attention from researchers worldwide. Quantum information technology represents a paradigm shift in information science, offering solutions to problems that have remained intractable with conventional processing methods.

Quantum computers leverage the unique properties of quantum entanglement and superposition to process data, thereby significantly enhancing computational speed and information processing efficiency. For instance, numerous calculations that require exponential time on classical computers can be performed in polynomial time on quantum computers. Using Shor's quantum algorithm on a functional quantum computer, the factorization of any 10,000-bit large integer can be accomplished in mere hours. In contrast, executing the same computation even on today's fastest supercomputer would take billions of years.

Furthermore, the storage capacity of quantum processors far exceeds that of existing computers. A quantum computer with 1,000 quantum bits can theoretically store information about all the elementary particles in the universe. This extraordinary storage capacity, combined with the computational power of quantum algorithms, positions quantum computing as a potential game-changer in various fields, including cryptography, material science, and drug discovery.

The uncertainty principle, a fundamental concept in quantum mechanics, dictates that the transmission of quantum states can be disrupted by various means, including quantum channel noise and enemy interference. Furthermore, the inherent probabilistic nature of quantum mechanics introduces the possibility of quantum errors, which can occur without detection. If these errors cannot be corrected, quantum communication will be compromised, rendering information unintelligible. Therefore, to safeguard quantum information, a process of quantum error correction coding is essential. This method, which employs redundancy in information transmission, is the quantum counterpart to traditional channel coding techniques.

In 1995, Peter Shor constructed the first true quantum code using repetition codes, namely the binary quantum $[[9,1,3]]$ code(\cite{Shor}). Subsequently, Calderbank, Rains, Shor, Sloane and Gottesman independently proposed methods for constructing quantum error correction codes from binary self-orthogonal codes(\cite{CRSS},\cite{GPHD}). This self-orthogonal can be symplectic self-orthogonal, Euclidean self-orthogonal, or quaternion Hermitian self-orthogonal. The construction of quantum error correction codes is then transformed into the study of self-orthogonal codes that meet the corresponding conditions and called CSS construction or CSS-like construction, and the constructed quantum code is a special type of stabilizer quantum code or additive quantum code.

Afterwards, various classical error correction codes, such as BCH codes(\cite{KKKS},\cite{LLX}), Reed Solomon codes(\cite{LXW2}), finite geometry codes(\cite{MLWZ}), algebraic geometry codes(\cite{CLX},\cite{J}), Turbo codes(\cite{PTO}), etc., can be used to construct quantum codes as long as they meet suitable self-orthogonal conditions. In addition, the application of Boolean functions and projection operators can simultaneously construct additive and non-additive quantum codes(\cite{AC}).

Constructing quantum MDS codes is an important topic, and a lot of research has been done on it. The quantum MDS codes of length $n\leqslant q+1$ have been constructed in \cite{GB},\cite{JLLX},\cite{RGT}. In recent years, lots of quantum codes with length $n$ bigger than $q+1$ have been constructed, but most of them have minimal distance smaller than $\frac{q}{2}+1$(\cite{JLLX}), so people are interested in the construction of quantum MDS codes with minimal distance bigger than $\frac{q}{2}+1$. Here we show some known results of construction of quantum MDS codes in Table $1$, Table $2$ and our results in Table $3$.

\begin{sidewaystable}[thp]
\renewcommand{\arraystretch}{2.0}
\caption{Some known results on quantum MDS codes with parameters $[[n,n-2d+2,d]]_q$.}
\centering
    \begin{tabular}{c|c|c|c} 
    \textbf{$q$} & \textbf{$n$} & \textbf{$d$} & \textbf{Reference}\\
      \hline
      $\sim$                    &       $n\leqslant q+1$        &       $2\leqslant d\leqslant\frac{n}{2}+1$                    &\cite{GB},\cite{JX.E},\cite{RGT}      \\
      $\sim$                    &       $n=q^2$                 &       $d\leqslant q+1$                                        &\cite{JLLX},\cite{LXW}     \\
      $\sim$                    &       $n=\frac{q^2}{2}$       &       $2\leqslant d\leqslant q$                               &\cite{KZL}      \\
      $\sim$                    &       $n=tq$, $1\leqslant t\leqslant q$                  &       $2\leqslant d\leqslant \lfloor\frac{tq+q-1}{q+1}\rfloor+1$                                        &\cite{FF.T},\cite{LXW},\cite{S}      \\
      $q+1\mid n-2$             &       $n=q^2+1$               &       $2\leqslant d\leqslant q+1$, $d \neq q$                 &\cite{B},\cite{FF.T},\cite{G},\cite{GR},\cite{JX.A},\cite{KZ},\cite{LXW}      \\
      $q+1\mid n-2$             &       $n=t(q+1)+2$            &       $d\leqslant t+2$, $(p,t,d)\neq(2,q-1,q)$                &\cite{FF.T}      \\
      even                      &       $n=q^2+1$               &       $d\leqslant q+1$, $d$ is odd                            &\cite{G}      \\
      $q\equiv1(mod\ 4)$        &       $n=q^2+1$               &       $d\leqslant q+1$, $d$ is even                           &\cite{KZ}      \\
      $m\mid q+1$,$m$ even      &       $\frac{q^2-1}{m}$       &       $2\leqslant d\leqslant\frac{q-1}{2}+\frac{q-1}{m}$      &\cite{WZ}      \\
      $m\mid q-1$,$m$ even      &       $\frac{q^2-1}{m}$       &       $2\leqslant d\leqslant\frac{q+1}{2}+\frac{q-1}{m}$      &\cite{CLZ}     \\
      $m\mid q+1$,$m$ odd       &       $\frac{q^2-1}{m}$       &       $2\leqslant d\leqslant\frac{q-1}{2}+\frac{q-1}{2m}$     &\cite{CLZ},\cite{WZ}       \\
      $m\mid q+1$,$m$ odd       &       $\frac{q^2+m-1}{m}$     &       $2\leqslant d\leqslant\frac{q-1}{2}+\frac{q-1}{2m}$     &\cite{HXC}     \\
      $s\mid q-1$               &       $1+r\frac{q^2-1}{s}$    &       $2\leqslant d\leqslant r\frac{q-1}{s}+1$     &\cite{FF.S}     \\
      $s\mid q-1$,$s$ even      &       $\frac{q^2-1}{s}$       &       $2\leqslant d\leqslant(\frac{s}{2}+1)\frac{q-1}{s}+1$     &\cite{GZ.S}     \\
      $2s\mid q+1$              &       $2t\frac{q^2-1}{2s}$, $1\leqslant t\leqslant s$    &     $2\leqslant d\leqslant(s+t)\frac{q-1}{2s}-1$     &\cite{GZ.S},\cite{SY}     \\
      $2s+1\mid q+1$            &       $2t\frac{q^2-1}{2s+1}$, $1\leqslant t\leqslant s-1$    &       $2\leqslant d\leqslant(s+t+1)\frac{q-1}{2s+1}-1$     &\cite{JKW},\cite{SY}     \\
      $2s+1\mid q+1$            &       $(2t+1)\frac{q^2-1}{2s+1}$, $1\leqslant t\leqslant s$    &       $2\leqslant d\leqslant(s+t+1)\frac{q-1}{2s+1}-1$     &\cite{FF.S},\cite{JKW}     \\
    \end{tabular}
\end{sidewaystable}

\begin{sidewaystable}[thp]
\renewcommand{\arraystretch}{2.0}
\caption{Some known results on quantum MDS codes with parameters $[[n,n-2d+2,d]]_q$.}
\centering
    \begin{tabular}{c|c|c|c} 
    \textbf{$q$} & \textbf{$n$} & \textbf{$d$} & \textbf{Reference}\\
      \hline
      $q>2,2s+1\mid q+1$        &       $1+r\frac{q^2-1}{2s+1}$, $1\leqslant r\leqslant 2s+1$ &       $2\leqslant d\leqslant(s+1)\frac{q-1}{2s+1}$     &\cite{FF.S}     \\
      $q>2,2s+1\mid q+1$        &       $1+(2t+1)\frac{q^2-1}{s}$, $0\leqslant t\leqslant s-1$     &       $2\leqslant d\leqslant(s+t+1)\frac{q+1}{2s}$     &\cite{FF.S}     \\
      $2s\mid q+1$       &       $1+r\frac{q^2-1}{2s}$, $2\leqslant r\leqslant 2s$     &    $2\leqslant d\leqslant(s+1)\frac{q+1}{2s}$     &\cite{FF.S} \\
      $2s\mid q+1$       &       $1+(2t+2)\frac{q^2-1}{2s}$, $0\leqslant t\leqslant s-2$   &    $2\leqslant d\leqslant(s+t+1)\frac{q+1}{2s}$     &\cite{FF.S}  \\
      $2s\mid q+1$       &       $(2t+1)\frac{q^2-1}{2s}$, $1\leqslant t\leqslant s-1$     &    $2\leqslant d\leqslant(s+t)\frac{q+1}{2s}-1$     &\cite{FF.S}  \\
      $s$ odd, $s\mid q-1,2s-1\mid q+1$  &     $\frac{q^2-1}{s}$     &     $2\leqslant d\leqslant \frac{q-1}{2}+\frac{q+1}{4s-2}$     &\cite{HXC}      \\
      $s$ even,$t$ even, $s\mid q+1,t\mid q-1$ & \makecell[c]{$n=r\frac{q^2-1}{s}+l\frac{q^2-1}{t}-2rl\frac{q^2-1}{st}$\\$1\leqslant r\leqslant s-1,t\geqslant2,1\leqslant l\leqslant t$}  &     $d\leqslant min\{r\frac{q+1}{s}-1,\frac{q+1}{2}+\frac{q-1}{t}-1\}$     &\cite{JLFQ}      \\
      $s$ even,$t$ even, $s+1\mid q+1,t\mid q-1$  &     \makecell[c]{$n=1+r\frac{q^2-1}{s+1}+l\frac{q^2-1}{t}-rl\frac{q^2-1}{(s+1)t}$\\$1\leqslant r\leqslant s,t\geqslant2,1\leqslant l\leqslant t$}     &     $d\leqslant min\{r\frac{q+1}{s+1},\frac{q+1}{2}+\frac{q-1}{t}-1\}$     &\cite{JLFQ}      \\
      $s$ odd,$t$ even, $s\mid q+1,t\mid q-1$  &     \makecell[c]{$n=h\frac{q^2-1}{s}+r\frac{q^2-1}{t}-hr\frac{q^2-1}{st}$\\
      $r\leqslant t,h\leqslant s-1,q-1>hr\frac{q^2-1}{st}$}     &     $d\leqslant min\{\lfloor\frac{s+h}{2}\rfloor\cdot\frac{q+1}{s}-1,\frac{q+1}{2}+\frac{q-1}{t}\}$     &\cite{FL}      \\
      $s$ odd,$t$ even, $s\mid q+1,t\mid q-1$  &     \makecell[c]{$n=1+h\frac{q^2-1}{s}+r\frac{q^2-1}{t}-hr\frac{q^2-1}{st}$\\$h$ odd,$r\leqslant t$}     &     $d\leqslant min\{\frac{s+h}{2}\frac{q+1}{s},\frac{q+1}{2}+\frac{q-1}{t}\}$     &\cite{FL}      \\
      $s$ even,$t$ even, $s\mid q+1,t\mid q-1$  &     \makecell[c]{$n=1+h\frac{q^2-1}{s}+r\frac{q^2-1}{t}$\\$h\leqslant \frac{s}{2},r\leqslant \frac{t}{2}$}     &     $d\leqslant min\{\lfloor\frac{s+h}{2}\rfloor\frac{q+1}{s}-1,\frac{q+1}{2}+\frac{q-1}{t}\}$     &\cite{FL}      \\
      $q$ odd                   &       $\frac{q^2+1}{2}$        &       $2\leqslant d\leqslant q$, $d$ odd                    &\cite{KZ}      \\
      $q\equiv\pm3(mod\ 10)$    &       $\frac{q^2+1}{5}$        &       $2\leqslant d\leqslant \frac{3q\pm1}{5}$, $d$ even    &\cite{GZ.Q},\cite{HYZ},\cite{KZL}      \\
      $q\equiv\pm2(mod\ 10)$    &       $\frac{q^2+1}{5}$        &       $2\leqslant d\leqslant \frac{3q\mp1}{5}$, $d$ odd     &\cite{LXG}      \\
    \end{tabular}
\end{sidewaystable}

Using CSS construction, we can construct quantum codes using self-orthogonal codes that satisfy certain conditions. If these self-orthogonal codes happen to be MDS codes, the constructed quantum code is a quantum MDS code(\cite{AK}). As a special type of MDS code, generalized Reed-Solomon codes have the characteristics of convenient construction and simple structure. Therefore, constructing quantum MDS by finding generalized Reed-Solomon codes that satisfy certain self-orthogonal conditions is a feasible construction scheme.

In this paper, we will construct two kinds of  quantum MDS codes with minimal distance larger than $\frac{q}{2}$. The paper is organized as follows. In section $2$, we introduce some preliminary knowledge on generalized Reed-Solomon codes which will be useful for the remaining paper. In section $3$, we will present our results obtained from multiplicative subgroups of $\F_{q^2}^*$. In section $4$, the construction is obtained from the $(q+1)$-th  root of unity in $\F_{q^2}$. In section $5$, we will make a conclusion and some further problems will be proposed.

\newpage
\begin{table}[thp]
\renewcommand{\arraystretch}{2.0}
\caption{Our results on quantum MDS codes with parameters $[[n,n-2d+2,d]]_q$.}
\centering
    \begin{tabular}{c|c|c} 
    \textbf{$q$} & \textbf{$n$} & \textbf{$d$} \\
      \hline
      $q=sm+1$ with $q,s$ odd        &       $\lambda\frac{q^2-1}{s}, 1\leqslant\lambda\leqslant s$   &       $d=k+1,k\leqslant\frac{s+1}{2}m$    \\
              $\sim$               &       $q^2-s(q+1)$, $s\leqslant\frac{q-1}{2}$                  &       $d=k+1,k\leqslant \frac{q^2-s(q+1)-2}{q+1}+1$
    \end{tabular}
\end{table}

\section{Hermitian consruction}
\qquad In this section, we introduce some basic notations and useful results on generalized Reed-Solomon codes (GRS codes for short). Let $g$ be a primitive element of $\F_{q^2}$ such that $\F_{q^2}^*=\F_{q^2}\setminus\{0\}=\langle g\rangle$. For $\alpha=(a_1, \cdots, a_n),\beta=(b_1, \cdots, b_n)\in \F_{q^2}^n$, denote the Hermitian inner product of $\alpha$ and $\beta$ by $\langle \alpha, \beta\rangle_h=a_1^q b_1+\cdots+a_n^q b_n$.

Let $\overrightarrow{\alpha}=(\alpha_1,\cdots, \alpha_n)$ and $\overrightarrow{v}=(v_1,\cdots, v_n)$ where $\alpha_i \in \F_{q^2}$ are distinct and $v_i\in \F_{q^2}^*$ for $1\leqslant i\leqslant n$. Then for $1\leqslant k\leqslant n$, we can define GRS code in the following:
$$
{\rm GRS}_k(\overrightarrow{\alpha},\overrightarrow{v})=\{(v_1f(\alpha_1),v_2f(\alpha_2),\cdots,v_nf(\alpha_n))\mid f(x)\in\F[x],\deg(f(x))<k\}.
$$
The code ${\rm GRS}_k(\overrightarrow{\alpha},\overrightarrow{v})$ is an MDS code with parameters $[n,k,n-k+1]_{q^2}$ which has a generator matrix(for further information on GRS codes, see Chapter One in\cite{MFSN})
$$
\begin{bmatrix}
v_1                       &    v_2                    &   \ldots  &   v_n                                      \\
v_1\alpha_1               &    v_2\alpha_2            &   \ldots  &   v_n\alpha_n                              \\
\vdots                    &    \vdots                 &   \ddots  &   \vdots                                   \\
v_1\alpha_1^{k-1}         &    v_2\alpha_2^{k-1}      &   \ldots  &   v_n\alpha_n^{k-1}                        \\
\end{bmatrix}
.$$

The guiding construction scheme is due to Ashikhmin and Knill in $2001$(\cite{AK}).

\begin{theorem}{\bf(Hermitian Construction)}
If there exists an $[n,k,d]_{q^2}$ linear code $C$ with $C^{\perp H}\subseteq C$, where $C^{\perp H}$ is Hermitian dual code of $C$, then there exist an $[[n,2k-n,\geqslant d]]_q$ quantum code.
\end{theorem}

Notice that the Hermitian dual of an MDS code is also MDS.

\begin{corollary}{\bf(Hermitian Construction for Quantum MDS Codes \rm{\cite{KKKS}})}\label{HCQ}
If there exists an $[n,d,n-d+1]_{q^2}$ MDS code $C$ with $C^{\perp H}\subseteq C$, where $C^{\perp H}$ is Hermitian dual code of $C$, then there exists an $[[n,n-2d,d+1]]_q$ quantum MDS code.
\end{corollary}

Besides, the following result on checking ${\rm GRS}_k(\overrightarrow{\alpha},\overrightarrow{v})$ Hermitian self-orthogonality will be useful.
\begin{theorem}
{\rm(\cite{R})}${\rm GRS}_k(\overrightarrow{\alpha},\overrightarrow{v})$ is Hermitian self-orthogonal iff $\langle \overrightarrow{\alpha}^{qi+j}, \overrightarrow{v}^{q+1}\rangle_h=0$ for all $0\leqslant i,j\leqslant k-1$.
\end{theorem}

\section{The case $q=sm+1$ for odd $s$ and $q$}

\qquad Here $n=\frac{q^2-1}{s}=m(sm+2)$.

\begin{lemma}\label{000}
   There does not exist $0\leqslant i,j\leqslant\frac{s+1}{2}m-1$ such that  $qj+j\equiv \frac{s+1}{2}m\pmod{n}$.
\end{lemma}

\begin{proof}
   Assume \[qj+j\equiv \frac{s+1}{2}m\pmod{n}\] for some $0\leqslant i,j\leqslant\frac{s+1}{2}m-1$.  Taking module $sm+2$ on both sides yields
   \[j-i\equiv \frac{s+1}{2}m\pmod{sm+2}\]
   which is a contradiction to $0\leqslant i,j\leqslant\frac{s+1}{2}m-1$.
\end{proof}

Let $H$ be a multiplicative subgroup of $\F_{q^2}^*$ of order $n=\frac{q^2-1}{s}$, i.e, $H=\langle g^s\rangle$. Now denote by $H=\{\alpha_1, \alpha_2, \cdots, \alpha_n\}$ with $\alpha_k=g^{sk}$ and $\mu=\frac{s+1}{2}m$.

\begin{theorem}\label{Main}
    There exists an Hermitian self-orthogonal GRS code with parameters $[n,k,n-k+1]_{q^2}$ where $k\leqslant\frac{s+1}{2}m$. Moreover, there exists quantum MDS code with parameters $[[n,n-2k, k+1]]_q$.
\end{theorem}

\begin{proof}

  Denote by $\overrightarrow{\alpha}=(\alpha_1,\cdots, \alpha_n)$ and $\mu=\frac{s+1}{2}m$. By Theorem 2, the main step is to find  $\overrightarrow{v}=(v_1,\cdots, v_n)$ satisfying $\langle \overrightarrow{\alpha}^{qi+j}, \overrightarrow{v}^{q+1}\rangle_h=0$ for any $0\leqslant i, j\leqslant k-1$. Then $\mathrm{GRS}_k(\overrightarrow{\alpha}, \overrightarrow{v})$ is Hermitian self-orthogonal.

    {\bf Firstly}, note that $\alpha_r^{-\mu}=g^{-rs\mu}=g^{q^2-1-rs\mu}$. We claim that for any $1\leqslant r\leqslant n$, ${\rm Tr}_{\F_{q^2}/\F_q}(\alpha_r^{-\mu})$ and ${\rm Tr}_{\F_{q^2}/\F_q}(g\alpha_r^{-\mu})$ can not be $0$ at the same time. Otherwise
    $${\rm Tr}_{\F_{q^2}/\F_q}(\alpha_r^{-\mu})={\rm Tr}_{\F_{q^2}/\F_q}(g\alpha_r^{-\mu})=0$$
    can imply
    $$g^{rs\mu(q-1)}=g^{(rs\mu-1)(q-1)}=-1.$$
    The $g^{q-1}=1$ which contradicts that $g$ is primitive in $\F_{q^2}$.

    {\bf Secondly}, we prove that there exists $t\in \F_{q^2}^*$ such that ${\rm Tr}_{\F_{q^2}/\F_q}(t\alpha_r^{-\mu})\neq 0$ for all $r$, $1\leqslant r\leqslant n$.

    If for every $1\leqslant r\leqslant n$, we have $\alpha^{-\mu}_r\neq -1$. Then when $t=1$, ${\rm Tr}_{\F_{q^2}/\F_q}(t\alpha_r^{-\mu})\neq 0$ for all $r$, $1\leqslant r\leqslant n$. Otherwise, we will demonstrate that when $t=g$, ${\rm Tr}_{\F_{q^2}/\F_q}(t\alpha_r^{-\mu})\neq 0$ as follows: suppose that there exists $1\leqslant r_0\leqslant n$ such that
    $$\alpha^{-\mu}_r=g^{r_0s\mu(q-1)}=-1.$$
    Therefore $\frac{q^2-1}{2}\mid r_0s\mu(q-1)$ which is equivalent to
    \begin{equation}
    (q+1)(2l'+1)=r_0s(s+1)m
    \end{equation}
    for some integer $l'$. If there exists $1\leqslant r_1 \leqslant n$ such that ${\rm Tr}_{\F_{q^2}/\F_q}(g\alpha_{r_1}^{-\mu})=0$, then
    $$g\alpha_{r_1}^{-\mu}=g^{(r_1s\mu-1)(q-1)}=-1,$$
    which is equivalent to
    \begin{equation}
    (q+1)(2l''+1)=r_1s(s+1)m-2
    \end{equation}
    with some integer $l''$. By subtracting (2) from (1), we can get
    $$2(q+1)(l'-l'')=2(r_0-r_1)s(s+1)m+2.$$ For $q=sm+1$ is odd, $sm$ must be even. Consequently, $s(sm+m)$ is also even. Therefore, the above equality does not hold, which leads to a contradiction. Therefore, for $1\leqslant r \leqslant n$ we have ${\rm Tr}_{\F_{q^2}/\F_q}(g\alpha_r^{-\mu})\neq0$.

     Denote by $\overrightarrow{v}=(v_1, \cdots, v_n)$ such that $v_r^{q+1}={\rm Tr}_{\F_{q^2}/\F_q}(g\alpha_r^{-\mu})$. Then Lemma \ref{000} implies $\langle \overrightarrow{\alpha}^{qi+j}, \overrightarrow{v}^{q+1}\rangle_h=0$ for any $0\leqslant i, j\leqslant k-1$. In this case, the code $\mathrm{GRS}_k(\overrightarrow{\alpha}, \overrightarrow{v})$ is Hermitian self-orthogonal with parameters $[n,k,n-k+1]_{q^2}$. Then applying Corollary \ref{HCQ} to the MDS code $\mathrm{GRS}_k(\overrightarrow{\alpha}, \overrightarrow{v})$, we get quantum MDS codes with parameters $[[n,n-2k, k+1]]_q$.
\end{proof}

\begin{corollary}\label{001}

    There exist quantum MDS codes with parameters $[[\lambda n,\lambda n-2k,k+1]]_q$, where $k\leqslant\frac{s+1}{2}m, 1\leqslant\lambda\leqslant s$.
\end{corollary}

\begin{proof}
  Choose $H=\{\alpha_1, \cdots, \alpha_n\}$ and $(v_1,\cdots, v_n)$ as in Theorem \ref{Main}. Choose distinct cosets $g_1H, g_2H, \cdots, g_\lambda H$ with $\lambda\leqslant s$.

    Let
    $$
    \oa=(g_1\alpha_1, g_1\alpha_2, \cdots, g_1\alpha_n, g_2\alpha_1, g_2\alpha_2, \cdots, g_2\alpha_n, \cdots,
    g_\lambda\alpha_1,g_\lambda\alpha_2,\cdots,g_\lambda\alpha_n)
    $$
    and
    $$
    \ov=(g_1v_1, g_1v_2, \cdots, g_1v_n, g_2v_1, g_2v_2, \cdots, g_2v_n, \cdots,g_\lambda v_1,g_\lambda v_2,\cdots,g_\lambda v_n).
    $$
   Then $\langle \overrightarrow{\alpha}^{qi+j}, \overrightarrow{v}^{q+1}\rangle_h=0$.  Hence we obtain Hermitian self-orthogonal GRS code ${\rm GRS}_k(\oa,\ov)$ with parameters $[\lambda n,k,\lambda n-k+1]_{q^2}$. As a result, quantum MDS codes with parameters $[[\lambda n, \lambda n-2k, k+1]]_q$ can be obtained.
\end{proof}

\begin{example}

    In the case $s=3,\lambda\leqslant3$ and $q\equiv3m+1$ being even, there exist quantum MDS codes with parameters
    $[[\lambda\frac{q^2-1}{3},\frac{\lambda q^2-4q+4-\lambda}{3},\frac{2q+1}{3}]]_q$ whose minimal distance are larger than $\frac{2q}{3}\approx0.667q$.
\end{example}

\begin{example}

    In the case $s=5,\lambda\leqslant5$ and $q\equiv5m+1$ being even, there exist quantum MDS codes with parameters
    $[[\lambda\frac{q^2-1}{5},\frac{\lambda q^2-6q+6-\lambda}{5},\frac{3q+2}{5}]]_q$ whose minimal distance are larger than $\frac{3q}{5}=0.6q$.
\end{example}

\section{The case $n=q^2-s(q+1)$}

\qquad In this section, denote $\F_{q^2}=\{\alpha_1,\alpha_2,\cdots,\alpha_{q^2-q-1},\alpha_{q^2-q},\alpha_{q^2-q+1},\cdots,\alpha_{q^2}\}$, where the last $q+1$ elements are in  $S=\langle g^{q-1}\rangle$. Then for $1\leqslant i\leqslant q^2-q-1$ we have
$$
\prod\limits_{j=1,j \ne i}^{q^2-q-1}(\alpha_i-\alpha_j)=\frac{\prod\limits_{j=1,j \ne i}^{q^2}(\alpha_i-\alpha_j)}{\prod\limits_{j=q^2-q}^{q^2}(\alpha_i-\alpha_j)}=
\frac{-1}{\alpha_i^{q+1}-1}\in \F_q.
$$

Then according to Lemma 5 in \cite{LXW}, when $qi+j\leqslant q^2-q-3$, there exists $v_i\in \F_{q^2}$ such that $v_i^{q+1}=\prod\limits_{j=1,j \ne i}^{q^2-q-1}(\alpha_i-\alpha_j)^{-1}$ such that
$$
\sum\limits_{s=1}^{q^2-q-1}\alpha_s^{qi+j}v_s^{q+1}=\sum\limits_{s=1}^{q^2-q-1}\frac{\alpha_s^{qi+j}}{\prod\limits_{k=1,k \ne s}^{q^2-q-1}(\alpha_s-\alpha_k)}=0.
$$

Set $\oa=\{\alpha_1,\alpha_2,\cdots,\alpha_{q^2-q-1}\}$ and $\ov=\{v_1,v_2,\cdots,v_{q^2-q-1}\}$. Then by Corollary \ref{HCQ}, ${\rm GRS}_k(\oa,\ov)$ is Hermitian self-orthogonal with $k\leqslant \frac{q^2-q-3}{q+1}+1$.

\begin{lemma}

    There exist quantum MDS codes with parameters $[[q^2-q-1,q^2-q-1-2k,k+1]]_q$, where $k\leqslant \frac{q^2-q-3}{q+1}+1$.
\end{lemma}

\begin{proof}

    Apply Corollary \ref{HCQ} to ${\rm GRS}_k(\oa,\ov)$ with parameters $[q^2-q-1,k,q^2-q-k]_{q^2}.$
\end{proof}

\begin{theorem}\label{qodd}

    If $q$ is odd prime power, there exist quantum MDS codes with parameters $[[q^2-s(q+1),q^2-s(q+1)-2k,k+1]]_q$, where $k\leqslant \frac{q^2-s(q+1)-2}{q+1}+1,s\leqslant \frac{q-1}{2}$.
\end{theorem}

\begin{proof}

    Since $S$ is the set of $(q+1)$-th unit roots in $\F_{q^2}$ while $\F_q^*$ is the set of $(q-1)$-th unit roots in $\F_{q^2}$,   the intersection of $S$ and $\F_q^*$ is $\{\pm1\}$. So the different cosets of $S$ with representative in $\F_q$ are $S,g^{q+1}S,g^{2(q+1)}S,\cdots,g^{(\frac{q-1}{2}-1)(q+1)}S$. Decompose $\F_q=\{\alpha_1,\alpha_2,\dots,\alpha_{\frac{q^2-1}{2}}\}\bigcup S\bigcup g^{q+1}S\bigcup\cdots\bigcup g^{(\frac{q-1}{2}-1)(q+1)}S$.

    Therefore, for $s\leqslant\frac{q-1}{2},k\leqslant\frac{q^2-s(q+1)-2}{q+1}+1$, let
    $$
    \begin{aligned} 
    \nonumber
        \oa=( & \alpha_1,\alpha_2,\dots,\alpha_{q^2-s(q+1)},g^{q-1},g^{2(q-1)},\cdots,g^{(q+1)(q-1)},g^{2q},g^{3q-1},\cdots,g^{q+1},\cdots,      \\
               & g^{(\frac{q-1}{2}-1-s)(q+1)}g^{q-1},g^{(\frac{q-1}{2}-1-s)(q+1)}g^{2(q-1)},\cdots,g^{(\frac{q-1}{2}-1-s)(q+1)}g^{(q+1)(q-1)})    \\
    \end{aligned}
    $$
    and $\ov=(v_1,v_2,\cdots,v_{q^2-s(q+1)})$ where
    $$
    \begin{aligned} 
    \nonumber
        v_i^{q+1} &= \prod\limits_{j=1,j \ne i}^{q^2-s(q+1)}(\alpha_i-\alpha_j)^{-1}           \\
                  &= -\prod\limits_{j=1,j\neq i}^{q+1}(\alpha_i-g^{j(q-1)})\prod\limits_{j=1,j\neq i}^{q+1}(\alpha_i-g^{q+1}g^{j(q-1)})\cdots
                  \prod\limits_{j=1,j\neq i}^{q+1}(\alpha_i-g^{(\frac{q-1}{2}-1-s)(q+1)}g^{j(q-1)})
    \end{aligned}
    $$
    belongs to $\F_q$(Then $v_i$ exists).

    Then, for $i,j\leqslant k-1$ we have \[\sum\limits_{s=1}^{q^2-s(q+1)}\alpha_s^{qi+j}v_s^{q+1}=\sum\limits_{s=1}^{q^2-s(q+1)}\frac{\alpha^{qi+j}_s}{\prod\limits_{j=1,j \ne i}^{q^2-s(q+1)}(\alpha_i-\alpha_j)}=0.\] Now ${\rm GRS}_k(\oa,\ov)$  is  an Hermitian self-orthogonal MDS code with parameters $[q^2-s(q+1),k,q^2-s(q+1)-k+1]_{q^2}$. So there exists quantum MDS code with parameters $[[q^2-s(q+1),q^2-s(q+1)-2k,k+1]]_q$.
\end{proof}

\begin{theorem}

    If $q$ is even, there exist quantum MDS codes with parameters $[[q^2-s(q+1),q^2-s(q+1)-2k,k+1]]_q$, where $k\leqslant q-1-s,s\leqslant q-1$.
\end{theorem}

\begin{proof}

    The proof is similar as Theorem \ref{qodd}. Notice that in this case the different cosets of $S$ with representative elements in $\F_q$ are precisely $S,g^{q+1}S,g^{2(q+1)}S,\cdots,g^{(q-2)(q+1)}S$.
\end{proof}

\begin{example}

    We have quantum MDS codes with parameters $[[q^2-s(q+1),q^2-(s+2)(q-1),q-s]]_q$ where $s\leqslant q-1$ whose minimal distance is approximately $q$ when $q$ is large and $s$ is small.
\end{example}

\section{Conclusion}

\qquad In this paper, we construct two new classes of quantum MDS codes by using Hermitian self-orthogonal GRS codes from some multiplicative subgroup of $\F_{q^2}^*$. Their minimum distances are large than $\frac{q}{2}$ in some cases. Moreover, the finite field $\F_{q^2}$ has abundant additive group structure as well as multiplicative group structure. More quantum MDS codes with larger minimal distances may be constructed if we can properly combine those group structures together.

\section{Acknowledgements}
The authors are supported by National Natural Science Foundation of China
(Nos. 12171191, 12271199) and by self-determined research funds of CCNU from the colleges’ basic research and operation of MOE CCNU22JC001.

\bibliographystyle{IEEEtran}
\bibliography{IEEEexample}
\end{document}